\documentclass[10pt,fleqn,twoside]{article}
\usepackage[cp1251]{inputenc}

\usepackage{amsfonts,amsmath,amssymb,amsthm,latexsym,amscd,mathrsfs}
\usepackage{ifthen,cite}
\usepackage[bookmarksnumbered=true]{hyperref}
\hypersetup{pdfpagetransition={Split}}

\newcommand{\evenhead}{Author \ name}
\newcommand{\oddhead}{Article \ name}
\newcommand{\theArticleName}{Article name}

\newcommand{\FirstPageHeading}[1]{\thispagestyle{empty}%
\noindent\raisebox{0pt}[0pt][0pt]{\makebox[\textwidth]{\protect\footnotesize \sf }}\par}

\newcommand{\ArticleName}[1]{\renewcommand{\theArticleName}{#1}\vspace{-2mm}
\par\noindent {\LARGE\bf  #1\par}}
\newcommand{\Author}[1]{\vspace{5mm}\par\noindent {\it #1} \par\vspace{2mm}\par}
\newcommand{\Address}[1]{\vspace{2mm}\par\noindent {\it #1} \par}
\newcommand{\Email}[1]{\ifthenelse{\equal{#1}{}}{}{\par\noindent {\rm E-mail: }{\it  #1} \par}}
\newcommand{\URLaddress}[1]{\ifthenelse{\equal{#1}{}}{}{\par\noindent {\rm URL: }{\tt  #1} \par}}
\newcommand{\EmailD}[1]{\ifthenelse{\equal{#1}{}}{}{\par\noindent {$\phantom{\dag}$~\rm E-mail: }{\it  #1} \par}}
\newcommand{\URLaddressD}[1]{\ifthenelse{\equal{#1}{}}{}{\par\noindent {$\phantom{\dag}$~\rm URL: }{\tt  #1} \par}}

\newcommand{\Keywords}[1]{\vspace{3mm}\par\noindent\hspace*{10mm}
\parbox{130mm}{\small {\it Key words:} \rm #1}\par}
\newcommand{\Classification}[1]{\vspace{3mm}\par\noindent\hspace*{10mm}
\parbox{130mm}{\small {\it 2000 Mathematics Subject Classification:} \rm #1}\vspace{3mm}\par}

\newcommand{\ShortArticleName}[1]{\renewcommand{\oddhead}{#1}}
\newcommand{\AuthorNameForHeading}[1]{\renewcommand{\evenhead}{#1}}

\long\def\@makecaption#1#2{
  \sbox\@tempboxa{\small \textbf{#1.}\ \ #2}%
  \ifdim \wd\@tempboxa >\hsize
    {\small \textbf{#1.}\ \ #2}\par \else
    \global \@minipagefalse
    \hb@xt@\hsize{\hfil\box\@tempboxa\hfil}%
  \fi \vskip\belowcaptionskip}

\def\numberwithin#1#2{\@ifundefined{c@#1}{\@nocounterr{#1}}{%
  \@ifundefined{c@#2}{\@nocnterr{#2}}{%
  \@addtoreset{#1}{#2}%
  \toks@\@xp\@xp\@xp{\csname the#1\endcsname}%
  \@xp\xdef\csname the#1\endcsname
    {\@xp\@nx\csname the#2\endcsname.\the\toks@}}}}

\newtheorem{lemma}{Lemma}

\newtheorem{proposition}{Proposition}

{
\theoremstyle{definition}

}

\def \E^#1{{ \buildrel #1 \over\vee}}

\begin{document}
\allowdisplaybreaks

\FirstPageHeading{Tatiana~V.~Ryabukha}
\ShortArticleName{On the existence of functionals for mean values of
observables}

\ArticleName{On the existence of functionals for the mean values \\of
observables
}

\Author{Tatiana~V.~Ryabukha}
\AuthorNameForHeading{T.~V.~Ryabukha}

\Address{Institute of Mathematics of NAS of Ukraine, Kyiv, Ukraine}


{\vspace{6mm}\par\noindent\hspace*{8mm}
\parbox{140mm}{\small { $\quad$
The aim of this work is to study the existence of mean values of observables
for infinite-particle systems.
Using solutions of the initial value problems to the BBGKY hierarchy and to its dual, we prove the local, in time, existence of the  mean value functionals in the cases where either observables or states vary in time.
We also discuss problems on the existence of such functionals for several different classes of observables and for an arbitrary time interval.

\Keywords{infinite-particle systems; 
BBGKY hierarchy;
dual BBGKY hierarchy;
cumulant (semi-invariant); mean values of observables.}
\Classification{35Q40; 47d06.}
}}}

\makeatletter
\renewcommand{\@evenhead}{
\hspace*{-3pt}\raisebox{-15pt}[\headheight][0pt]{\vbox{\hbox to \textwidth {\thepage \hfil \evenhead}\vskip4pt \hrule}}}
\renewcommand{\@oddhead}{
\hspace*{-3pt}\raisebox{-15pt}[\headheight][0pt]{\vbox{\hbox to \textwidth {\oddhead \hfil \thepage}\vskip4pt\hrule}}}
\renewcommand{\@evenfoot}{}
\renewcommand{\@oddfoot}{}
\makeatother

\newpage
\renewcommand{\contentsname}{Content}
\protect \tableofcontents
\thispagestyle{empty}

\newpage
\section{Introduction}
\setcounter{equation}{0}
\underline{}
A particular progress in the study of dynamics of infinitely many particles have been achieved during the last decade \cite{CGP97,CIP94,Pe08,Spohn04}. As is well known the evolution of states of such systems is completely determined by an initial value problem to the BBGKY hierarchy \cite{CGP97}. A solution of
this initial value problem represented in the form of a series expansion as a result of integration of perturbation series with respect to the time variables was first constructed in \cite{Pe79} for a one-dimensional system of particles interacting via short range potential with hard-core for initial data close to equilibrium states. The divergence of integrals with respect to the configuration variables in every term of the expansion is a typical problem which complicates the construction of the solution for infinite-particle systems. The method of an interaction region \cite{PG83,Pe79,CGP97} provides one of the approaches to eliminate this obstacle. A further result on the infinite-particle dynamics for three-dimensional hard sphere systems \cite{CIP94,GP85,Lan75,PG90,Spohn04} is based on the construction of a solution of the BBGKY hierarchy in the form of perturbation series.

Recently a solution of the initial value problem to the BBGKY hierarchy  \cite{GR02umj,GRS04physA} has been constructed in the form of expansion over particle clusters, the evolution of which is governed by the corresponding order cumulant of the evolution operators of finitely many particles. The above-mentioned representations for a solution of the BBGKY hierarchy are particular cases of those from \cite{GR02umj,GRS04physA}. The method of an interaction region  for this representations is developed in  \cite{R06SIGMA}.

In this paper  we develop a similar approach to the regularization of a solution of the dual BBGKY hierarchy   \cite{GR02umj,GR03} which describes the evolution of marginal observables of many-particle systems \cite{BG01,GR03}. We prove the existence of the mean values for infinite-particle  systems in both cases where either the evolutions of observables or the evolution of states is considered.

We first present some preliminary facts about description of infinite-particle systems. We consider a one-dimensional system of identical particles (with unit mass)
interacting via a short range pair potential
$\Phi$
(with  hard core) that possesses the following properties
\begin{equation}
\label{PhiSpher}
    \begin{array}{ll}
            a)&    \Phi\in C^2([\sigma,R]), \quad 0<\sigma<R<\infty,\\
            b)&    \Phi(|q|)=\left\{\begin{array}{ll}
                                    +\infty, \quad &|q|\in[0,\sigma),\\
                                    0,             &|q|\in(R,\infty),
                                    \end{array}
                             \right.\\
            c)&    \Phi'(\sigma+0)=0,
    \end{array}
\end{equation}
\par\noindent
where
$R$
is the radius  of forces acting between particles with the length
$\sigma>0$.

Every $i$th particle is characterized by a coordinate
$q_i$
of the center of hard-core and momentum
$p_i$.
Let us denote
$x_{i}\equiv(q_{i},p_{i})\in\mathbb{R}\times{\mathbb R},\; i\geq1.$
For
$n$-particle system the inequalities
\begin{equation*}
    |q_i-q_j|\geq\sigma,
    \quad
    i\neq j\in\{1,\ldots,n\},
    \quad
    n\geq2,
\end{equation*}
hold, i.~e.  particles can occupy only admissible configurations. The set
$W_n
\equiv
\big\{
    (q_1,\ldots,q_n)\in\mathbb{R}^n
    \mathrel|
    \text{for at least one pair }
    (i,j),\, i\neq j\in\{1,\ldots,n\},\,
    \text{such that }
    |q_i-q_j|<\sigma,
\big\}$,
$n\geq2,$
is a region of forbid\-den configurations.

At the initial instant
$t=0$
observables are described by the sequences of marginal ($s$-particle) functions
$G(0)=(G_0,G_{1}(0,x_1),\ldots,G_{s}(0,x_1,\ldots,x_s),\ldots)$
and states by the sequences of marginal ($s$-particle) distribution functions
$F(0)=\big(F_0,F_{1}(0,x_{1}),\ldots,F_{s}(0,x_{1},\ldots,x_{s}),\ldots\big)$.
A mean value (a mathematical expectation) of observable
$G(0)$
for a system in the state
$F(0)$
is determined by the following functional
\cite{CGP97,BG01}:
\begin{equation}
    \label{<G>}
    \langle G \rangle (0)
   =\langle G(0)|F(0) \rangle
   =\sum\limits_{s=0}^{\infty}
        \frac{1}{s!}
        \int\limits_{(\mathbb{R}^s\setminus W_s)\times\mathbb{R}^s}
            dx_1\ldots dx_s\,
            G_s(0,x_1,\ldots,x_s)\,
            F_s(0,x_1,\ldots,x_s),
\end{equation}
where
$F_0=1,$ $G_0=0.$

To describe infinite particle systems we introduce the space
$C_\gamma$
of sequences
$g=\big(g_0,g_1(x_1),\ldots\linebreak\ldots,g_n(x_1,\ldots,x_n),\ldots\big)$
of bounded (continuous) functions
$g_n (x_1,\ldots, x_n)$, $n\geq0$
($g_0$ is a number), given on the phase space
$(\mathbb{R}^n\setminus W_n)\times\mathbb{R}^n$,
symmetric with respect to arbitrary permutations of the arguments
$x_i,\;i=1,\ldots,n,$
and equal zero in the region of forbid\-den configurations
$W_n$,
with the norm \cite{BG01}
\begin{equation*}
\label{Cgamma}
    \| g\|_C{_\gamma}
   =\sup_{n\geq0}\,
    \frac{\gamma^n}{n!}\,
    \sup_{x_1,\ldots, x_n}
    \big|
        g_n(x_1,\ldots, x_n)
    \big|,
\end{equation*}
\par\noindent
where
$0<\gamma <1$
is a constant.
We denote  by
$C_{\gamma,0}\subset C_\gamma$
the subspace of finite sequences of continuously differentiable
functions with compact supports on the configuration space.
The sequence of functions
$g\in C_{\gamma,0}$
is treated as a quasiobservable (analog an local observable
\cite{CGP97,Bo48}).

States of infinite particle systems are described by sequences from the space
$L_{\xi,\beta}^\infty$
of sequences
$f=\linebreak=\big(f_0,f_1(x_1),\ldots,f_n({x}_{1},\dots,{x}_{n}),\ldots\big)$
of functions
$f_n({x}_{1},\dots,{x}_{n}),$ $n\geq0$ $(f_0=1)$,
given on the phase space
$(\mathbb{R}^n\setminus W_n)\times\mathbb{R}^n$,
symmetric with respect to arbitrary permutations of the arguments
$x_i,\,i=1,\ldots,n$,
and equal zero on the forbidden configurations
$W_n$,
with the norm \cite{CGP97,Pe79}
\begin{equation*}
\label{Lxi,beta}
    ||f||_{L_{\xi,\beta}^\infty}
   =\sup_{n\geq0}\xi^{-n}\,
    \sup_{{x}_{1},\dots,{x}_{n}}
    \big|
        f_n({x}_{1},\dots,{x}_{n})
    \big|\,
    \exp\bigg\{\beta\sum\limits_{i=1}^n\frac{p_i^2}{2}\bigg\},
\end{equation*}
where
$\xi,\beta>0$
are constants.

We make note if
$F(0)\in L_{\xi,\beta}^\infty$
and
$G(0)\in C_{\gamma,0}$,
the following estimate
    \begin{equation*}
    \label{maj}
         \big|\langle G(0)|F(0) \rangle\big|
    \leq ||G(0)||_{C_\gamma}\,
         ||F(0)||_{L_{\xi,\beta}^\infty}\,
         \sum_{s=0}^\infty
         \bigg(\frac{C\,\xi}{\gamma}\sqrt{\frac{2\pi}{\beta}}\bigg)^s
    \end{equation*}
holds, therefore functional (\ref{<G>}) is well-defined under the condition
\begin{equation*}
\label{xi0<}
    \xi<\frac{\gamma}{C}\sqrt{\frac\beta{2\pi}},
\end{equation*}
where
$C=\max_{i=1,\ldots,s}|l_{i}(0)|$ 
and
$|l_{i}(0)|$
is a length of the interval
$ l_i(0)$
such that
$\Omega_s(0)=\linebreak={l_1(0)\times\ldots\times l_s(0)}$
\label{A}
is a support of function
$G_{s}(0)$
in the configuration space.

Let us note that for observable of the additive type $G^{(1)}(0)=\big(0,a_{1}(0,x_1),0,\ldots,0,\ldots\big)$
functional (\ref{<G>}) has the form
\begin{equation}
\label{<G1>}
    \big\langle
            G^{(1)}
    \big\rangle(0)
   =\big\langle
            G^{(1)}(0)|F(0)
    \big\rangle
   =\int\limits_{\mathbb{R}\times\mathbb{R}}
        dx_1\,
        a_1(0,x_1)\,
        F_1(0,x_1).
\end{equation}
If
$F(0)\in L_{\xi,\beta}^\infty$
and
$G(0)\in C_{\gamma,0}$,
then functional \eqref{<G1>} is well-defined for an arbitrary value of the parameter
$\xi>0$
since
\begin{equation*}
     \big|
        \langle
            G^{(1)}(0)|F(0)
        \rangle
     \big|
\leq ||G^{(1)}(0)||_{C_\gamma}\,
     ||F(0)||_{L_{\xi,\beta}^\infty}\,
     \frac{C\xi}{\gamma}\sqrt{\frac{2\pi}{\beta}}
    <\infty.
\end{equation*}

At an arbitrary instant of time
$t\in\mathbb{R}^1$
the mean value of observable is determined by the following functional
\cite{CGP97}:
\begin{equation}
    \label{<Gt>b}
    \langle G \rangle (t)
   =\langle G(t)|F(0) \rangle
   =\sum\limits_{s=0}^{\infty}
        \frac{1}{s!}
        \int\limits_{(\mathbb{R}^s\setminus W_s)\times\mathbb{R}^s}
            dx_1\ldots dx_s\,
            G_s(t,x_1,\ldots,x_s)\,
            F_s(0,x_1,\ldots,x_s),
\end{equation}
where
$G(t)=\big(0,G_1 (t, x_1),\ldots,G_s (t,x_1,\ldots, x_s),\ldots\big)$
is a solution of the initial value problem to the dual BBGKY hierarchy
\cite{BG01,GR03} with the initial data
$G(0)$,
or by the other functional:
\begin{equation}
    \label{<Gt>a}
    \langle G \rangle (t)
   =\langle G(0)|F(t) \rangle
   =\sum\limits_{s=0}^{\infty}
        \frac{1}{s!}
        \int\limits_{(\mathbb{R}^s\setminus W_s)\times\mathbb{R}^s}
            dx_1\ldots dx_s\,
            G_s(0,x_1,\ldots,x_s)\,
            F_s(t,x_1,\ldots,x_s),
\end{equation}
where
$F(t)= \big( 1, F_1 (t, x_1),\ldots,F_s (t,x_1,\ldots, x_s),\ldots \big)$
is a solution of the initial value problem to the BBGKY hierarchy
\cite{Bo48,CGP97,Pe08} with the initial data
$F(0)$.

Let us outline the structure of the paper.
In the second section we apply the procedure of the regularization  suggested in \cite{R06SIGMA} to a solution of the dual BBGKY hierarchy.
Then, in Sections~3~and~4, we prove the existence of  functionals (\ref{<Gt>b}) and (\ref{<Gt>a}) using results of \cite{R06SIGMA}. In the final section we give some concluding remarks.

\medskip
\section{The regularized solution of the dual BBGKY hierarchy}
\setcounter{equation}{0}
Before proving the existence of functional (\ref{<Gt>b}) we cite preliminary facts about a solution of the dual BBGKY hierarchy and construct a regularized representation for the solution which in the next section allows us to compensate the divergent terms appearing in functional (\ref{<Gt>b}) for infinite-particle systems.

Let
$(x_1,\ldots,x_s)=Y,$
$(x_{j_1},\ldots,x_{j_{s-n}})=Y\setminus X$,
$\{j_1,\ldots,j_{s-n}\}\subseteqq\{1,\ldots,s\}$,
e.~i.
$X=(x_1,\overset{\E^{j_1}}{},\ldots,\overset{\E^{j_{s-n}}}{},x_s)$,
where
$(x_1,\overset{\E^{j_k}}{\ldots},x_s)
 \equiv
 (x_1,\ldots,x_{j_{k-1}},x_{j_{k+1}},\ldots,x_s)$.
We denote by
$|X|$
a number of elements of the set
$X$, $0\leq |X|=n\leq s$.

We introduce the evolution operator
$S_{s}(t,Y),$
$s\geq1$,
defined on the space
$C_\gamma$
of sequences of continuous functions
by the following formula \cite{CGP97}
\begin{equation}
\label{Sspher}
    S_{|Y|}(t,Y)\,
    f_{|Y|}(Y)
   =\begin{cases}
         f_{|Y|}\big(\mathbf{X}_{1}(t,Y),\ldots,\mathbf{X}_{|Y|}(t,Y)\big),
         &Y\in
         \big(
              (\mathbb{R}^{|Y|}\setminus W_{|Y|})
              \times
              \mathbb{R}^{|Y|}
         \big)
         \setminus
         {\mathcal{M}}_{|Y|}^0,\\
         0,
         &Y\in\big(W_{|Y|}\times\mathbb{R}^{|Y|}\big)
         \cup{\mathcal{M}}_{|Y|}^0,
    \end{cases}
\end{equation}
\medskip
where
$\mathbf{X}_{j}(t,Y),$
$j=1,\ldots,|Y|,$
is the phase trajectory \cite{CGP97} of a system of
$|Y|=s$
particles with the initial data
$\mathbf{X}_{j}(0,Y)=x_j$
and
$S_{|Y|}(0)=I$
is a unit operator. We note that the phase trajectory of a system  \cite{CGP97} with the interaction potential (\ref{PhiSpher}) is defined not for all initial data
$(x_1,\dots,x_n)\in(\mathbb{R}^n\setminus W_n)\times\mathbb{R}^n$,
it is defined almost everywhere on the phase space
$(\mathbb{R}^n\setminus W_n)\times\mathbb{R}^n$,
namely, exterior to a certain set
${\mathcal{M}}_{n}^0$
of the Lebesgue measure zero (the set
${\mathcal{M}}_{n}^0$
contains the initial data
$(x_1,\dots,x_n)
 \in(\mathbb{R}^n\setminus W_n)\times\mathbb{R}^n$
for which:
\textit{i})
triple and more order particle collisions occur at the instant
$t\in(-\infty,+\infty)$;
\textit{ii})
the infinite number of collisions occurs within a finite time interval
\cite{CGP97,PG90}).
The evolution operator (\ref{Sspher}) is well defined as
$t\in(-\infty,+\infty)$
\cite{CGP97} under the conditions (\ref{PhiSpher}) on the
interaction potential
$\Phi$.

For the initial data
$G(0)\in C_\gamma$
in case of
$\gamma<e^{-1}$
a solution of the initial value problem to the dual BBGKY hierarchy is
represented  by the formula \cite{GR02umj,GRS04physA}
\begin{equation}
\label{Gstcum}
    G_{|Y|} (t, Y)
   =\sum\limits_{n=0}^s
        \sum\limits_{1=j_1<\ldots<j_{s-n}}^s\,
            \mathfrak{A}_{|X|+1}(t,Y\setminus X,\,X)\,
            G_{|Y\setminus X|}(0,Y\setminus X),
    \quad
    1\leq |X|=n\leq s,
\end{equation}
where the operator
$\mathfrak{A}_{|X|+1}(t,Y\setminus X,\,X)$
is a cumulant of $(|X|+1)$th order for the evolution operators
$S_{|Y_i|}(t,Y_i)$ (\ref{Sspher}),
$|Y_i|\geq1$,
determined by the formula \cite{GR02umj,GRS04physA}
\begin{equation}
\label{gU+}
    \mathfrak{A}_{|X|+1}(t,Y\setminus X,X)
   =\sum\limits_{\mathbf{P}:\,\{Y\setminus X,\,X\}=\bigcup\limits_iY_i}
        (-1)^{|\mathbf{P}|-1}\,
        (|\mathbf{P}|-1)!\,
        \prod_{Y_i\subset \mathbf{P}}
            S_{|Y_i|}(t,Y_i),
\end{equation}
where
$\sum\limits_\mathbf{P}$
is a sum over all possible partitions
$\mathbf{P}$
of the set
$Y\equiv\{Y\setminus X,\,X\}$
into
$|\mathbf{P}|$
nonempty mutually disjoint subsets
$Y_i\subset Y$,
$Y_i\!\bigcap\! Y_j\!=\!{\O},$ $i\neq j,$
and the subset
$Y\setminus X$
is treated as one element, i.~e.
$|Y|=|X|+1=n+1$. 
Let us define
$\mathfrak{A}_{|X|+1}(t,Y\setminus X,X)=0$
as
$Y=X$.
For example, the 1st and 2nd order cumulants have the form, correspondingly:
\addtocounter{equation}{-1}
\begin{subequations}
    \begin{align}
    \label{gU+a}
        &\mathfrak{A}_{1}(t,Y)
        =S_{s}(t,Y),               \\[3pt]
    \label{gU+b}
        &\mathfrak{A}_{2}(t,Y\setminus x_s,\, x_s)
        =S_{s}(t,Y)
        -S_{s-1}(t,Y\setminus x_s) \,
         S_1(t,x_s),
    \end{align}
\end{subequations}
where
$Y\setminus x_s\equiv\{x_1 \cup\ldots\cup x_{s-1}\}$
is treated as one element.

We note the cumulant can be defined on the set with elements being
finitely many-particle disjoint sets. For instance, the cumulant of 2nd
order defined on the two-element set with two elements as disjoint sets
$Y\setminus X$
and
$Z$
such that
$Y\setminus X\bigcap Z={\O}$
has the form
\begin{equation}
\label{gU+2}
         \mathfrak{A}_{2}(t,Y\setminus X,\, Z)
        =S_{|Y\setminus X|+|Z|}(t,Y\setminus X, Z)
        -S_{|Y\setminus X|}(t,Y\setminus X) \,
         S_{|Z|}(t,Z).$$
\end{equation}

\noindent
\begin{lemma}
\label{L:reg+}
   For the cumulant of $(n+1)$th order
   $\mathfrak{A}_{|X|+1}(t,Y\setminus X,\,X)$
   (\ref{gU+})
   the following representation is true \cite{GR02umj}:
    \begin{multline}
    \label{gU12+}
        \mathfrak{A}_{|X|+1}(t,Y\setminus X,\,X)
       =\sum\limits_{\substack{Z\subset X\\ Z\neq{\O}}}\,
            \mathfrak{A}_{2}(t,Y\setminus X,Z)
            \sum\limits_{\mathbf{Q}:\,X\setminus Z=\bigcup\limits_lX_l}
                (-1)^{|\mathbf{Q}|}\,
                |\mathbf{Q}|!\,
                \prod_{X_l\subset \mathbf{Q}}
                    \mathfrak{A}_{1}(t,X_l),\\
        \quad
        2\leq|X|=n\leq s,
    \end{multline}
    where
    $\sum\limits_{Z}$
    is a sum over all nonempty subsets
    $Z$
    of the set
    $X,$
    $\sum\limits_{\mathbf{Q}}$
    is a sum over all partitions
    $\mathbf{Q}$
    of the set
    $X\setminus Z$
    into
    $|\mathbf{Q}|$
    nonempty disjoint subsets
    $X_l\subset X\setminus Z,$
    $X_k\cap X_l=\O,\, k\neq l.$
\end{lemma}

We apply Lemma \ref{L:reg+} in constructing a new representation for the
solution of the initial value problem to the dual BBGKY hierarchy.

With representation (\ref{gU12+}) for solution (\ref{Gstcum}) in mind
the function
$G_{|Y\setminus X|}(0,Y\setminus X)$
proved to not depend on the set of variables
$X_l\subset X\setminus Z$.
Since
$X_l\subset\!\!\!\!\!\!/\, Y\setminus X$,
then the cumulants of the 1st order
$\mathfrak{A}_{1}(t,X_l)$
do not act on the variables of the function
$G_{|Y\setminus X|}(0,Y\setminus X)$.
Therefore an identical expression to (\ref{Gstcum}) is as follows:
\begin{multline*}
\label{Gst120}
    G_{|Y|} (t, Y)
   =\mathfrak{A}_{1}(t,Y)\,
    G_{|Y|}(0,Y)+\\
   +\sum\limits_{n=1}^s\,
    \sum\limits_{1=j_1<\ldots<j_{s-n}}^s\,
    \sum\limits_{\substack{Z\subset X\\ Z\neq{\O}}}\,
        \mathfrak{A}_{2}(t,Y\setminus X,Z)\,
        G_{|Y\setminus X|}(0,Y\setminus X)
        \sum\limits_{\mathbf{Q}:\,X\setminus Z=\bigcup\limits_lX_l}\,
            (-1)^{|\mathbf{Q}|}\,
            |\mathbf{Q}|!,\\
    \quad
    1\leq |X|=n\leq s.
\end{multline*}
\noindent

From last one, due to  the equalities
\begin{equation*}
    \sum\limits_{\mathbf{Q}:\,X\setminus Z=\bigcup\limits_lX_l}
        (-1)^{|\mathbf{Q}|}\,
        |\mathbf{Q}|!
   =\sum\limits_{k=1}^{|X\setminus Z|}
        (-1)^k k!\,
        s(|X\setminus Z|,k),
\end{equation*}
and
 \begin{equation*}
 \label{(-1)^m}
    \sum\limits_{k=1}^m\,
    (-1)^k\,
    k!\,
    s(m,k)=
    (-1)^m,
    \quad
    m\geq1,
\end{equation*}
where
$s(|X\setminus Z|,k)\equiv s(m,k)$
are Stirling numbers of the second kind,
it follows that an equivalent representation to solution (\ref{Gstcum}) of
the initial value problem to the dual BBGKY hierarchy is an expansion
({\it a regularized solution})
\begin{multline}
\label{Greg}
    G_{|Y|} (t, Y)
   =\mathfrak{A}_{1}(t,Y)\,
    G_{|Y|}(0,Y)
   +\sum\limits_{n=1}^s\,
    \sum\limits_{1=j_1<\ldots<j_{s-n}}^s\,
    \sum\limits_{\substack{Z\subset X\\ Z\neq{\O}}}\,
        (-1)^{|X\setminus Z|}\,
        \mathfrak{A}_{2}(t,Y\setminus X,Z)\,
        G_{|Y\setminus X|}(0,Y\setminus X),\\
    \quad
    1\leq |X|=n\leq|Y|=s.
\end{multline}

Thus, for
$G(0)\in C_{\gamma,0}$
cumulant representation (\ref{Gstcum})-(\ref{gU+}) for a solution of the
initial value problem to the BBGKY hierarchy is equivalent to regularized
one (\ref{Greg}) and its structure allows to compensate divergent terms appearing in functional (\ref{<Gt>b}) under
$F(0)\in L_{\xi,\beta}^\infty$.

\noindent
\begin{lemma}
\label{L:reg+<=}
    If
    $G(0)\in C_\gamma,$
    then under condition
    $0<\gamma <1$
    for expansion (\ref{Greg}) the estimate
    \begin{equation}
    \label{|G|<}
        \big|
            {G_{|Y|}(t,Y)}
        \big|
   \leq 2\,e^2\,
        \|G(0)\|_{C_{\gamma}}\,
        \frac{s!}{\gamma^{s}}
   \end{equation}
   is valid.
\end{lemma}

\noindent
\begin{proof}
    Let
    $G(0)\in C_{\gamma}.$
    According to formulae (\ref{Sspher}) and (\ref{gU+a}), (\ref{gU+b}) the
    following inequalities hold:
    \begin{equation}
    \label{gU+1<}
        \big|\mathfrak{A}_{1}(t,Y)\,
            G_{|Y|}(0,Y)
        \big|
        \leq
        \|G(0)\|_{C_{\gamma}}\,
        \frac{|Y|!}{\gamma^{|Y|}}
    \end{equation}
    and
    \begin{equation}
    \label{gU+2<}
        \big|\mathfrak{A}_{2}(t,Y\setminus X,Z)\,
            G_{|Y\setminus X|}(0,Y\setminus X)
        \big|
        \leq
        2\,\|G(0)\|_{C_{\gamma}}\,
        \frac{|Y\setminus X|!}{\gamma^{|Y\setminus X|}}.
    \end{equation}

    From (\ref{gU+1<}) and (\ref{gU+2<}) for expansion (\ref{Greg}) we get the
    following estimate
    \begin{equation}
    \label{|G|<<}
        \big|
            {G_{|Y|}(t,Y)}
        \big|
   \leq 2\,
        ||G(0)||_{C_\gamma}\,
        \sum\limits_{n=0}^s\,
            \sum\limits_{1=j_1<\ldots<j_{s-n}}^s\,
                2^{n}\,
                \frac{(s-n)!}{\gamma^{s-n}}.
    \end{equation}
    Since
    $0<\gamma<1$,
    $\sum\limits_{1=j_1<\ldots<j_{s-n}}^s
     1\,
    =\frac{s!}{(s-n)!n!}
    $
    and in consequence of the inequality
       $\sum\limits_{n=0}^{s}
        \frac{2^n}{n!}
        \leq
        e^2$,
    estimate (\ref{|G|<<}) takes form (\ref{|G|<}) of Lemma \ref{L:reg+<=}.
\end{proof}

\medskip
\section{The mean value of observables: the evolution of observables}
\setcounter{equation}{0}
\setcounter{proposition}{0}
We establish the existence of functional (\ref{<Gt>b}) in the case of the evolution of the observable determined by the regularized expansion (\ref{Greg}).
\begin{proposition}
\label{P:<Gt|F0>}
    For a system of particles with a pair interaction potential
    $\Phi$
    satisfying (\ref{PhiSpher}) if \linebreak
    $F(0)\in L_{\xi,\beta}^\infty$,
    $G(0)\in C_{\gamma,0}$ 
    and obserbvable 
    $G(t)$ 
    is determined by expansions (\ref{Greg}), then for
    \begin{equation}
    \label{xi}
        \xi
       <\frac{\gamma}{e\,\max{(C;2\tilde{C_1})}}\,
        \sqrt{\frac{\beta''}{2\pi}}
        \qquad
        \text{and}
        \qquad
        0\leq t<t_0,
        \quad
        t_0
        \equiv
        \frac{1}{\tilde{C_2}}\,
        \bigg(-\tilde{C_1}
              +\frac{\gamma}{2e\,\xi}\sqrt{\frac{\beta''}{2\pi}}
        \bigg),
    \end{equation}
    functional (\ref{<Gt>b}) is well-defined  and the following estimate holds:
    \begin{equation}
    \label{b<TH}
        \Big|\langle G(t)|F(0)\rangle\Big|
       \leq 2\,
        e^{C}\,
        ||{F(0)}||_{L_{\xi,\beta}^{\infty}}\,
        ||G(0)||_{C_\gamma}\,
        \sum\limits_{s=0}^{\infty}
            \bigg(\frac{e\,C\,\xi}{\gamma}\,\sqrt{\frac{2\pi}{\beta''}}
            \bigg)^s\,
        \sum\limits_{n=0}^{\infty}
            \bigg(\frac{2\,e\,\xi}{\gamma}\,\sqrt{\frac{2\pi}{\beta''}}
            \bigg)^n\,
        \big(
            \tilde{C}_1+\tilde{C}_2t
        \big)^n,
    \end{equation}
    where
    $C=\max_{i=j_1,\ldots,j_{s-n}}\big|l_{i}(0)\big|$
    and
    $\big|l_{i}(0)\big|$
    is a length of the interval
    $l_i(0)$
    from compact
    $\Omega_{|Y\setminus X|}(0)=
     \linebreak
    =l_{j_1}(0)
     \times\ldots\times
     l_{j_{|Y\setminus X|}}(0)$
    on which the function
    $G_{s-n}(0)$
    is supported,
    $\tilde{C_1}=\max(2R,1),$
    \linebreak
    $\tilde{C_2}=
     \max\big(2(4b+1),\frac{2}{\beta^{\prime}}\big),$
    $\beta=
     \beta^\prime+\beta^{\prime\prime},$
    $b\equiv
     \underset{q\in[\sigma,R]}
     {\sup}\big|\Phi(q)\big|
     \big(\big[\frac{R}{\sigma}\big]\big)$
    and
    $\big[\frac{R}{\sigma}\big]$
    is an integer part of the number~$\frac{R}{\sigma}$.
\end{proposition}

\begin{proof}
    Let
    $F(0)\in L_{\xi,\beta}^\infty$,
    $G(0)\in C_{\gamma,0}$,
    then the observable
    $G_{|Y\setminus X|}(0,Y\setminus X)$
    as
    $Y\setminus X=(x_{j_1},\ldots,x_{j_{s-n}})$ \linebreak[4]
    is supported in the configuration space on the compact set which we denote by \linebreak[4]
    $\Omega_{|Y\setminus X|}(0)=
    {l_{j_1}(0)\times\ldots\times l_{j_{|Y\setminus X|}}(0)},$
    where
    $l_{i}(0)$
    segment with the length
    $|l_{i}(0)|$
    such that
    \begin{math}
    q_i\in l_{i}(0),
    \linebreak[4]
    i=j_1,\ldots,j_{s-n}.
    \end{math}
    If
    $n=0$
    the compact comes as
    $\Omega_{|Y|}(0)={l_1(0)\times\ldots\times l_{|Y|}(0)}$
    (see p.~\pageref{A}).

    If
    $G_s(t,x_1,\ldots,x_s)
     \equiv
     G_{|Y|}(t, Y)$
    is a solution of the dual BBGKY hierarhy \cite{R06SIGMA}
    determined by expansion (\ref{Greg}), then the expression for
    functional (\ref{<Gt>b}) has the form
    \begin{multline}
        \label{reg<Gt>b}
        \langle G(t)|F(0) \rangle
       =\sum\limits_{s=0}^{\infty}\,
            \frac{1}{s!}
            \int\limits_{(\mathbb{R}^s\setminus W_s)\times\mathbb{R}^s}
            dY\,
            \Big(\mathfrak{A}_{1}(t,Y)\,
                G_{|Y|}(0,Y)+\\
            +\sum\limits_{n=1}^s\,
                    \sum\limits_{1=j_1<\ldots<j_{s-n}}^s \,
                        \sum\limits_{\substack{Z\subset X\\ Z\neq{\O}}}\,
                            (-1)^{|X\setminus Z|}\,
                            \mathfrak{A}_{2}(t,Y\!\setminus X,Z)\,
                            G_{|Y\setminus X|}(0,Y\!\setminus X)
            \Big)\,
            F_{|Y|}(0,Y),\\
        \qquad
        1\leq |X|=n\leq|Y|=s.
    \end{multline}

    In function (\ref{reg<Gt>b}) the region
    $\mathbb{R}^s\setminus W_s$
    of integration with respect to the configuration variables is restricted to a domain within which the integrand is finite and nonzero.
    Since the 1st order cumulant
    $\mathfrak{A}_{1}(t,Y)$
    is the  operator
    $S_{|Y|}(t,Y)$ (\ref{Sspher}),
    therefore in expansion (\ref{reg<Gt>b}) the function \linebreak[4]
    $\mathfrak{A}_{1}(t,Y)\,G_{|Y|}(0,Y)\,F_{|Y|}(0,Y)$
    is integrated not over the whole configuration space, and only with respect
    to variables of such a compact
    $\Omega_{|Y|}(0)\subset\mathbb{R}^s\setminus W_s$,
    that is shifted along the configuration trajectory with the finite volume
    $V_{\Omega_{|Y|}}(0)
    ={l_{|Y|}(0)}^s$.

    The expression
    \begin{equation*}
        \sum\limits_{n=1}^s\,
            \sum\limits_{1=j_1<\ldots<j_{s-n}}^s\,
                \sum\limits_{\substack{Z\subset X\\ Z\neq{\O}}}\,
                    (-1)^{|X\setminus Z|}\,
                    \mathfrak{A}_{2}(t,Y\setminus X,Z)\,
                    G_{|Y\setminus X|}(0,Y\setminus X)\,
                    F_{|Y|}(0,Y)
    \end{equation*}
    equals zero if within time interval
    $[0,t)$
    particles with arbitrary initial data
    $X$
    from
    $\mathbb{R}^n\setminus W_n$
    don't interact with particles with fixed initial data
    $Y\setminus X$,
    since  in this case for the 2nd order cumulant
    $\mathfrak{A}_{2}(t,Y\setminus X,Z)$ (\ref{gU+2})
    as
    $Z\subset X$,
    $Z\neq{\O}$,
    the following property
    \begin{equation}
    \label{gU+G=0}
        \mathfrak{A}_{2}(t,Y\setminus X,Z)\,
        G_{|Y\setminus X|}(0,Y\setminus X)
       =0
    \end{equation}
    \medskip
    is true.
    Owing to the finiteness of potential
    $\Phi$ (\ref{PhiSpher})
    within time interval
    $[0,t)$
    interaction (elastic collision) can not occur between particles at
    a distance more than can be accomplished by these particles in total.
    By this means, in expansion (\ref{reg<Gt>b}) the region of integration
    over the configuration variables is bounded by {\it an~interaction~region}
    $\Omega_{|X|}\big(t,\Omega_{|Y\setminus X|}(0)\big)$
    of the particles with arbitrary initial data
    $X$
    with the particles with fixed initial data
    $Y\setminus X\in\Omega_{|Y\setminus X|}(0)$,
    within time interval
    $[0,t)$
    and has a finite volume
    $V_{\Omega_{|X|}}(t)$.
    We note that for all
    $Z\subset X$
    it is true
    $\Omega_{|Z|}(t)\subset\Omega_{|X|}(t)$.

    Indeed, since at the initial moment
    $t=0$
    an arbitrary $i$th particle with the phase coordinate from
    $Y\setminus X$,
    $i=j_1,\ldots,j_{|Y\setminus X|},$
    has a fixed values of momentum
    $p_i$
    and the localized configuration coordinate
    $q_i$,
    then at an arbitrary instant
    $\tau$, $\tau\in[0,t)$,
    momentum
    $p_i(\tau,Y\setminus X)$
    is determined from the initial value problem to the Hamilton
    equations \cite{CGP97} for a system of the finite number
    $|Y\setminus X|$
    of particles.   Therefore within
    $[0,t)$
    such $i$th particle accomplishes a distance
    $\int\limits_0^td\tau\,\big|p_i(\tau,Y\setminus X)\big|$.
    In total
    $|Y\setminus X|$
    particles accomplish the following one
    $\sum\limits_{x_i\in Y\setminus X}
          \int\limits_0^t
          d\tau\,
          \big|p_i(\tau,Y\setminus X)\big|$,
    and
    $|Y|=s$
    particles do
    $\sum\limits_{i=1}^s
          \int\limits_0^t
          d\tau\,
          \big|p_i(\tau,Y)\big|$.
    Let us suppose that the interval
    $l_{|Y\setminus X|}(0)$
    is such that
    $\big|l_{|Y\setminus X|}(0)\big|=
     \max_{i=j_1,\ldots,j_{|Y\setminus X|}}{\big|l_{i}(0)\big|}$.
    Taking into account the range of interaction
    $R$ we single out orderly
    $|Y|$
    segments with the length
    $R+\int\limits_0^td\tau\,\big|p_i(\tau,Y)\big|,$
    $i=j_1,\ldots,j_{s}$,
    to the left and to the right of the interval
    $l_{|Y\setminus X|}(0)$.
    We obtain the interval denoted by
    $l_{|X|}(t)$
    with the length
    $\big|l_{|X|}(t)\big|$
    such that
    \begin{equation*}
     \underbrace{l_{|X|}(t)\times\ldots\times l_{|X|}(t)}\limits_n
     \times
     \Omega_{|Y\setminus X|}(0)
    =\Omega_{|X|}\big(t,\Omega_{|Y\setminus X|}(0)\big)
    \end{equation*}
    and
    \begin{equation}
    \label{lt<}
        \big|l_{|X|}(t)\big|
        \leq
        \big|l_{|Y\setminus X|}(0)\big|
       +2sR
       +2\sum\limits_{i=1}^s
        \int\limits_0^t
        d\tau\,
        \big|p_i(\tau,Y)\big|.
    \end{equation}

    For the pair interaction potential
    $\Phi$
    satisfying conditions (\ref{PhiSpher}) the following inequality is fulfilled:
    \begin{equation}
    \label{Phi<=}
        \Big|\sum_{i<j=1}^{s}\Phi(q_i-q_j)\Big|\leq bs,
    \end{equation}
    where
    $b\equiv \underset{q\in[\sigma,R]}{\sup}\big|
     \Phi(q)\big|\big(\big[\frac{R}{\sigma}\big]\big),$
    $\big[\frac{R}{\sigma}\big]$
    is an integer part of the number    
    $\frac{R}{\sigma}$.

    Because of the inequality
    \begin{equation*}
        2|p_i(\tau)|\leq p_i^2(\tau)+1,
    \end{equation*}
    for an arbitrary value of momentum
    $p_i(\tau)\equiv p_i(\tau,Y)$
    the law of conservation of energy
    \begin{equation*}
        \sum\limits_{i=1}^{s}\,
            \frac{p_i^{2}}{2}\,
        \sum\limits_{i<j=1}^{s}\,
            \Phi(q_i-q_j)
       =\sum\limits_{i=1}^{s}\,
            \frac{p_i^{2}(\tau)}{2}\,
        \sum\limits_{i<j=1}^{s}\,
            \Phi\big(q_i(\tau)-q_j(\tau)\big)
    \end{equation*}
    and conditions (\ref{Phi<=}) imply upper boundedness of the sum of
    momenta at the instant
    $\tau$:
    \begin{equation}
    \label{p<}
       2\sum\limits_{i=1}^{s}\,
            p_i(\tau)
       \leq
        \sum\limits_{i=1}^{s}\,
            p_i^2
           +(4b+1)\,
            s,
    \end{equation}
    where
    $b$
    is determined by condition (\ref{Phi<=}).

    In consideration of (\ref{p<}), from (\ref{lt<}) we derive
    \begin{equation}
    \label{lt<<}
         \big|l_{|X|}(t)\big|
         \leq
         \big|l_{|Y\setminus X|}(0)\big|
        +2s\big(R+(4b+1)t\big)+
         t\sum\limits_{i=1}^s
         p_i^2.
    \end{equation}

    Let us assume
    $C=\max_{i=j_1,\ldots,j_{s-n}}|l_{i}(0)|$,
    $C_1\equiv 2R$,
    $C_2\equiv2(4b+1)$
    and
    $b$
    is determined by condition (\ref{Phi<=})
    then volume
    $V_{\Omega_{|X|}}(t)$
    of the interaction region
    $\Omega_{|X|}\big(t,\Omega_{|Y\setminus X|}(0)\big)$
    is finite, namely,
    \begin{equation}
    \label{V+<=}
        V_{\Omega_{|X|}}(t)
        \leq\Big(C
           +(C_1+C_2t)s
           +t\sum_{i=1}^{s}{p_i^2}
           \Big)^n
        C^{s-n}.
    \end{equation}

    Thus equality (\ref{gU+G=0}) is satisfied in $s$th term of series
    (\ref{reg<Gt>b}) over integration with respect to the configuration variables
    located exterior to the interaction region
    $\Omega_{|X|}(t,\Omega_{|Y\setminus X|}(0))$.
    As to the structure
    $\Omega_{|Y|}(0)
     \equiv
     \Omega_{|X|}(t,\Omega_{|Y\setminus X|}(0))$
    as
    $|X|\equiv n=0$.

    Taking into  account the cumulant property (\ref{gU+G=0}), by
    Lemma~\ref{L:reg+<=} and inequality~(\ref{V+<=}) for functional~(\ref{reg<Gt>b})
    the following estimate holds:
    \begin{multline}
        \label{b<}
           \Big|\langle G(t)|F(0)\rangle
           \Big|
          \leq 2\,
           ||{F(0)}||_{L_{\xi,\beta}^{\infty}}\,
           ||G(0)||_{C_\gamma}\,
           \sum\limits_{s=0}^{\infty}\,
           \frac{\xi^s}{\gamma^s}\,
           \int\limits_{\mathbb{R}^s}
                dp_1\ldots dp_s\,
                \exp\Big\{-\beta
                           \sum\limits_{i=1}^{s}\,
                           \frac{p_i^2}{2}
                    \Big\}\,
           \sum\limits_{n=0}^{s}\,
                \frac{2^n}{n!}
           \times\\
           \times
           \Big(C
                +(C_1+C_2t)\,s
                +t\sum_{i=1}^{s}{p_i^2}
           \Big)^n\,
           C^{s-n}.
    \end{multline}

    In view of the relation
    \begin{equation*}
            \Big(C
                +(C_1+C_2t)\,s
                +t\sum\limits_{i=1}^{s}\,
                      p_i^2
            \Big)^{n}
           =\sum_{k=0}^n\,
               \frac{n!\,C^k}{k!}\,
            \sum_{r=0}^{n-k}\,
               \frac{s^r}{r!}\,
               (C_1+C_2t)^r\,
               \frac{t^{n-k-r}}{(n-k-r)!}\,
            \Big(\sum\limits_{i=1}^{s}
                     p_i^2
            \Big)^{n-k-r}
    \end{equation*}
    and the inequality
    \begin{equation}
    \label{pexp<=}
        \Big(\sum\limits_{i=1}^s\,
                p_i^2
        \Big)^{n-k-r}\,
        \exp\Big\{-\beta^\prime\,
                  \sum\limits_{i=1}^s\,
                      \frac{p_i^2}{2}
            \Big\}
        \leq
        (n-k-r)!
        \left(\frac{2}{\beta^\prime}\right)^{n-k-r}
    \end{equation}
    calculating the integrals over momentum variables in expression (\ref{b<})
    we get
    \begin{multline}
    \label{b<<}
        \Big|
            \langle
                G(t)|F(0)
            \rangle
        \Big|
        \leq 2\,
        ||F(0)||_{L_{\xi,\beta}^{\infty}}\,
        ||G(0)||_{C_\gamma}\,
        \sum\limits_{s=0}^{\infty}\,
            \bigg(
                \frac{C\xi}{\gamma}
            \bigg)^s\,
            \bigg(
                \frac{2\pi}{\beta''}
            \bigg)^{\frac{s}{2}}\,
        \times\\
        \times
        \sum\limits_{n=0}^{s}\,
            \bigg(
                \frac{2}{C}
            \bigg)^n\,
        \sum_{k=0}^n\,
            \frac{C^k}{k!}\,
        \sum_{r=0}^{n-k}\,
            \frac{s^r}{r!}\,
            \big(
                C_1+C_2t
            \big)^r\,
            \bigg(
                \frac{2t}{\beta'}
            \bigg)^{n-k-r},
    \end{multline}
    where
    $\beta=\beta'+\beta''$.

    We assume
    $\tilde{C_1}=\max(C_1,1),$
    $\tilde{C_2}=\max\big(C_2,\frac{2}{\beta^{\prime}}\big).$
    Then for arbitrary
    $t>0$
    the following inequalities
    \begin{equation*}
        \tilde{C_1}+\tilde{C_2t}\geq1\quad
        \text{and}\quad
        \big(
            \tilde{C_1}+\tilde{C_2}t
        \big)
        \frac{\beta^{\prime}}{2t}\geq1
    \end{equation*}
    are satisfied. Because of these we have
    \begin{equation}
    \label{C1+C2t}
        \big(
            C_1+C_2t
        \big)^r
        \bigg(
            \frac{2t}{\beta^{\prime}}
        \bigg)^{n-k-r}
        \leq
        \big(
            \tilde{C_1}+\tilde{C_2}t
        \big)^n.
    \end{equation}
    Considering (\ref{C1+C2t}) and inequalities
    \begin{equation}
    \label{e^<}
        \sum_{k=0}^n\,
            \frac{C^k}{k!}
        \leq
        \mathrm{e}^C,
        \quad
        \sum_{r=0}^{n-k}\,
        \frac{s^r}{r!}
        \leq
        \mathrm{e}^s
    \end{equation}
    estimate (\ref{b<<}) takes the form
    \begin{equation}
    \label{b<<<}
        \Big|
            \left\langle
                G(t)|F(0)
            \right\rangle
        \Big|
        \leq
        2e^C\,
        ||{F(0)}||_{L_{\xi,\beta}^{\infty}}\,
        ||G(0)||_{C_\gamma}\,
        \sum\limits_{s=0}^{\infty}\,
        \bigg(
            \frac{e\,C\,\xi}{\gamma}
        \bigg)^s\,
        \bigg(
            \frac{2\pi}{\beta''}
        \bigg)^{\frac{s}{2}}\,
        \sum\limits_{n=0}^{s}\,
        \bigg(
            \frac{2}{C}
        \bigg)^n\,
        \big(
            \tilde{C_1}+\tilde{C_2t}
        \big)^n.
    \end{equation}

    After inverting the order of summation in (\ref{b<<<}) we definitely
    obtain estimate (\ref{b<TH}).

    Thus far, from inequality (\ref{b<TH}) it follows under condition
    (\ref{xi}) functional (\ref{reg<Gt>b}) for the mean value of
    observables is well-defined, with corresponding to functional (\ref{<Gt>b}) in the case of the regularized representation for a solution of the initial value problem to the dual BBGKY hierarchy.
\end{proof}

We note that regularized expansion (\ref{Greg}) for initial observables of
additive type
$G^{(1)}(0)=\big(0,a_{1}(0,x_1),\linebreak0,\ldots,0,\ldots\big)$
has the form
\begin{equation*}
\label{G(1)t}
    G^{(1)}(t)
   =(0,G^{(1)}_1(t,x_1),G^{(1)}_2(t,x_1,x_2),\ldots,
    G^{(1)}_s(t,x_1,\ldots,x_s),\ldots)
\end{equation*}
where
\addtocounter{equation}{-0}
\begin{subequations}
    \begin{align}
    \label{g(1)1}
        &G^{(1)}_1(t,x_1)
        =\mathfrak{A}_{1}(t,x_1)a_{1}(0,x_1),\\[10pt]
         \hskip-20pt
         \text{and}
        &\nonumber\\
    \label{g(1)s}
        &G^{(1)}_s(t,x_1,\ldots,x_s)
        =\sum\limits_{j=1}^s\,
         \sum\limits_{\substack{Z\subset Y\setminus x_j\\ Z\neq{\O}}}\,
         (-1)^{|Y\setminus Z|-1}\,
         \mathfrak{A}_{2}(t,x_j,Z)\,
         a_1(0,x_j),
         \quad
         s\geq2.
    \end{align}
\end{subequations}

If
$F(0)\in L_{\xi,\beta}^\infty$,
$G^{(1)}(0)\in C_{\gamma,0}$
and
$G^{(1)}_s(t,x_1,\ldots,x_s)$
is determined by formulae (\ref{g(1)1}), (\ref{g(1)s}) then functional
(\ref{<Gt>b}) takes the form
\begin{multline}
\label{G1tb}
       \big\langle
            G^{(1)}
       \big\rangle(t)
      =\big\langle
            G^{(1)}(t)|F(0)
       \big\rangle
      =\int\limits_{\mathbb{R}^{1}\times\mathbb{R}^{1}}
           dx_1\,
           \mathfrak{A}_{1}(t,x_1)\,
           a_{1}(0,x_1)\,
           F_{s}(0,x_1,\ldots,x_{s})+\\
      +\sum\limits_{s=2}^{\infty}\,
           \frac{1}{s!}\,
           \int\limits_{(\mathbb{R}^{s}\setminus W_{s})\times\mathbb{R}^{s}}
               dx_1\ldots dx_{s}\,
               \sum\limits_{j=1}^s\,
                   \sum\limits_{\substack{Z\subset Y\setminus x_j\\ Z\neq{\O}}}\,
                       (-1)^{|Y\setminus Z|-1}\,
                       \mathfrak{A}_{2}(t,x_j,Z)\,
                       a_1(0,x_j)\,
                       F_{s}(0,x_1,\ldots,x_{s})
\end{multline}
and the following estimate holds:
\begin{equation}
\label{G1t<b}
     \Big|
        \big\langle
             G^{(1)}(0)|F(t)
        \big\rangle
     \Big|
     \leq
     \frac{2\,e^{C+1}\,C\,\xi}{\gamma}\,
     \sqrt{\frac{2\pi}{\beta''}}\;
     ||{F(0)}||_{L_{\xi,\beta}^{\infty}}\,
     ||G^{(1)}(0)||_{C_\gamma}\,
     \sum\limits_{s=0}^{\infty}\,
         \big(2e\xi\big)^s
         \bigg(
            \frac{2\pi}{\beta''}
         \bigg)^\frac{s}2\,
     \big(
        \tilde{C}_1+\tilde{C}_2t
     \big)^{s}.
\end{equation}

Inequality (\ref{G1t<b}) is similar to (\ref{b<<<}).
Indeed, let
$F(0)\in L_{\xi,\beta}^\infty$
and
$G^{(1)}(0)\in C_{\gamma,0}.$
For (\ref{G1tb}) we get

\begin{equation}
\label{1tb<}
      \int\limits_{\mathbb{R}\times\mathbb{R}}
          dx_1\,
          \big|
            \mathfrak{A}^+_{1}(t,x_1)\,
            a_{1}(0,x_1)
          \big|\,
          \big|
            F_{1}(0,x_1)
          \big|
       \leq
       \frac{\xi\,C(t)}{\gamma}\,
       ||{F(0)}||_{L_{\xi,\beta}^{\infty}}\,
       ||G^{(1)}(0)||_{C_\gamma},
\end{equation}
where
$C(t)\equiv\big|l_1(0)\big|\sqrt{\frac{2\pi}{\beta}}+\frac{2t}{\beta}$,
and
\begin{multline}
\label{2tb<}
       \sum\limits_{s=2}^{\infty}\,
           \frac{1}{s!}\,
           \int\limits_{(\Omega_1(0)\times\Omega_{s-1}(t))\times\mathbb{R}^{s}}
               dx_1\ldots dx_{s}\,
               \sum\limits_{j=1}^s\,
                   \sum\limits_{\substack{Z\subset \{x_1,\ldots,x_{s}\}\setminus x_j\\Z\neq{\O}}}\,
                       \big|
                            \mathfrak{A}_{2}(t,x_j,Z)\,
                            a_{1}(0,x_j)
                       \big|\,
                       \big|
                            F_{s}(0,x_1,\ldots,x_{s})
                       \big|
       \leq\\
       \leq
       \frac{C\xi}{\gamma}\,
       e^{C+1}\,
       \sqrt{\frac{2\pi}{\beta''}}\;
       ||{F(0)}||_{L_{\xi,\beta}^{\infty}}\,
       ||G^{(1)}(0)||_{C_\gamma}\,
       \sum\limits_{s=1}^{\infty}\,
       (2e\xi)^s\,
       \bigg(
            \frac{2\pi}{\beta''}
       \bigg)^{\frac{s}{2}}\,
       \big(
            \tilde{C}_1+\tilde{C}_2t
       \big)^{s},
\end{multline}
where notations of estimate (\ref{b<TH}) are in service. According
to (\ref{1tb<}) and (\ref{2tb<}) we obtain estimate (\ref{G1t<b}).

It is notable that the existence of functionals for observables of $s$-fold
class has also been considered  for another representation of the solution to
the dual BBGKY hierarchy formulated in \cite{BG01}.
\medskip

\section{The mean value of observables: the evolution of states}
\setcounter{equation}{0}
\setcounter{proposition}{0}

We establish  the existence of functional (\ref{<Gt>a}), if the evolution
of states is described by the regularized representation \cite{R06SIGMA} for a
solution of the initial value problem to the BBGKY hierarchy.

Let
$(x_1,\ldots,x_s)\equiv Y,$
$(Y,x_{s+1},\ldots,x_{s+n})\equiv X,$
i.~e.
$X\setminus Y=(x_{s+1},\ldots,x_{s+n}),$
$dx_{s+1}\ldots dx_{s+n}\equiv\linebreak\equiv d(X\setminus Y)$.
We denote by
$|X|$
the number of elements of  the set
$X$, i.e. $|X|=|Y|+|{X\setminus Y}|=s+n$.

Let us consider functional (\ref{<Gt>a}) with the state
$F_s(t,x_1,\ldots,x_s)\equiv F_{|Y|}(t, Y)$
determined by a solution of the
initial value problem to the BBGKY hierarchy {\it (a regularized solution)}
\cite{R06SIGMA}
\begin{multline}
\label{FgU2}
       F_{|Y|}(t,Y)
      =\mathfrak{A}_{1}(-t,Y)\,
       F_{|Y|}(0,Y)+\\
      +\sum\limits_{n=1}^{\infty}\,
           \frac{1}{n!}
           \int\limits_{(\mathbb{R}^n\setminus W_n)\times\mathbb{R}^n}
               d(X\setminus Y)\,
               \sum\limits_{\substack{Z\subset X\setminus Y\\ Z\neq{\O}}}\,
                   (-1)^{|X\setminus(Y\cup Z)|}\,
                   \mathfrak{A}_{2}(-t,Y,Z)\,
       F_{|X|}(0,X),
       \quad
       |X\setminus Y|\geq1,
\end{multline}
where
$\sum\limits_{\substack{Z\subset X\setminus Y\\ Z\neq{\O}}}\!\!\!\!\!$
is a sum over all nonempty subsets
$Z$
of the set
$X\setminus Y$,
the evolution operator
$\mathfrak{A}_{1}(-t,Y)$
is a 1st order cumulant of the evolution operators (\ref{Sspher}):
\begin{equation*}
\label{U_1}
    \mathfrak{A}_{1}(-t,Y)=S_{|Y|}(-t,Y),
\end{equation*}
the evolution operator
$\mathfrak{A}_{2}(-t,Y,Z)$
is a 2nd order cumulant:
\begin{equation*}
\label{U_2}
    \mathfrak{A}_{2}(-t,Y,Z)
   =S_{|Y\cup Z|}(-t,Y,Z)
   -S_{|Y|}(-t,Y)
    S_{|Z|}(-t,Z).
\end{equation*}

\noindent
\begin{proposition}
\label{P:<G0|Ft>}
    For a system of particles with a pair interaction potential
    $\Phi$
    satisfying (\ref{PhiSpher}), if \linebreak
    $F(0)\in L_{\xi,\beta}^\infty$,
    $G(0)\in C_{\gamma,0}$
    and state is determined by expansion (\ref{FgU2}), then for
    \begin{equation}
    \label{xi<}
        \xi
       <\min\bigg(\frac{\gamma}{C};
                  \frac{2}{(2\tilde{C_1}+1)^2-1}
            \bigg)\,
        \mathrm{e}^{-2\beta b-1}
        \sqrt{\frac{\beta^{\prime\prime}}{2\pi}}
    \end{equation}
    and
    \begin{equation}
    \label{t0}
       0\leq t<t_0 \equiv\frac{1}{2\tilde{C_2}}
        \Bigg(
        -2\tilde{C_1}-1
        +\bigg(1+\frac{2\mathrm{e}^{-2\beta b-1}}{\xi}
                 \sqrt{\frac{\beta^{\prime\prime}}{2\pi}}
         \bigg)^{\frac12}
        \Bigg),
    \end{equation}
    functional (\ref{<Gt>a}) is well-defined and the following
    estimate holds:
    \begin{multline}
    \label{A<}
        \Big|
            \langle
                G(0)|F(t)
            \rangle
        \Big|
        \leq
        2\,e^C\,
        ||{F(0)}||_{L_{\xi,\beta}^{\infty}}\,
        ||G(0)||_{C_\gamma}\,
        \sum\limits_{s=0}^{\infty}\,
        \left(\frac{C\,\xi\,\mathrm{e}^{2\beta b+1}}{\gamma}
              \sqrt{\frac{2\pi}{\beta^{\prime\prime}}}
        \right)^s
        \times\\
        \times
        \sum_{n=0}^\infty
        \left(2\,\xi\,
              \mathrm{e}^{2\beta b+1}
              \big(
                \tilde{C_1}+\tilde{C_2}t
              \big)
              \big(
                1+\tilde{C_1}+\tilde{C_2}t
              \big)
              \sqrt{\frac{2\pi}{\beta^{\prime\prime}}}
        \right)^n,
    \end{multline}
    where
    $C=\max_{i=1,\ldots,s}|l_{i}(0)|$,
    $\tilde{C_1}=\max(2R,1),$
    $\tilde{C_2}=\max\big(2(4b+1),\frac{2}{\beta^{\prime}}\big)$
    and
    $\beta=\beta^\prime+\beta^{\prime\prime}.$
\end{proposition}

\begin{proof}
    Let
    $G(0)\in C_{\gamma,0}$
    and
    $F(0)\in L_{\xi,\beta}^{\infty}$.
    We suppose at the initial instant
    $t=0$
    data of the set
    $Y$
    are fixed in such a way that observable
    $G_{|Y|}(0, Y)$
    is located in the configuration space on the compact
    $\Omega_{|Y|}(0)={l_1(0)\times\ldots\times l_{|Y|}(0)}$
    (see p.~\pageref{A}).

    If the marginal distribution functions
    $F_s(t,x_1,\ldots,x_s)
    \equiv
    F_{|Y|}(t,Y)$
    in functional (\ref{<Gt>a}) are determined by expansion (\ref{FgU2}).
    Then functional (\ref{<Gt>a}) takes the form
    \begin{multline}
    \label{reg<Gt>a}
        \langle G(0)|F(t) \rangle
       =\sum\limits_{s=0}^{\infty}\,
            \frac{1}{s!}
            \int\limits_{(\mathbb{R}^s\setminus W_s)\times\mathbb{R}^s}
                dY\,
                \Big(
                    \mathfrak{A}_{1}(-t,Y)\,
                    F_{|Y|}(0,Y)+\\
                   +\sum\limits_{n=1}^{\infty}\,
                        \frac{1}{n!}
                        \int\limits_{(\mathbb{R}^n\setminus W_n)\times\mathbb{R}^n}
                            d(X\setminus Y){}\,
                            \sum\limits_{\substack{Z\subset X\setminus Y\\ Z\neq{\O}}}\,
                                (-1)^{|X\setminus(Y\cup Z)|}\,
                                \mathfrak{A}_{2}(-t,Y,Z)\,
                    F_{|X|}(0,X)
                \Big)\,
            G_{|Y|}(0, Y).
    \end{multline}
    \medskip

    In expansion (\ref{reg<Gt>a}) we restrict the region
    $\mathbb{R}^s\setminus W_s$
    of integration with respect to the configuration variables since the
    integrand is finite and nonzero. Indeed, taking into account the 1st order cumulant
    $\mathfrak{A}_{1}(-t,Y)$
    is equivalent to  operator
    $S_{|Y|}(-t,Y)$ (\ref{Sspher}),
    therefore in expansion (\ref{reg<Gt>a}) the expression
    $\mathfrak{A}_{1}(-t,Y)\,F_{|Y|}(0,Y)\,G_{|Y|}(0,Y)$
    is integrated not over the whole configuration space, and only with respect to variables of the compact
    $\Omega_{|Y|}(0),$
    $\Omega_{|Y|}(0)\subset\mathbb{R}^s\setminus W_s$,
    shifted along the configuration trajectory with a finite volume
    $V_{\Omega_{|Y|}}(0)
    ={l_{|Y|}}^s(0)$.

    The expression
    $$\sum\limits_{\substack{Z\subset X\setminus Y\\ Z\neq{\O}}}\,
     (-1)^{|X\setminus(Y\cup Z)|}\,
     \mathfrak{A}_{2}(-t,Y,Z)\,
     F_{|X|}(0,X)$$
    equals zero exterior to the region
    $\Omega_{|X\setminus Y|}(t,\Omega_{|Y|}(0))$
    of interaction of the particles with arbitrary initial data
    $X\setminus Y$
    with the particles with fixed initial data
    $Y$,
    $Y\in\Omega_{|Y|}(0)$
    within the time interval
    $[0,t)$
    as a consequence of the equality
    \begin{equation*}
        \mathfrak{A}_{2}(-t,Y,Z)\,
        F_{|X|}(0,X)
       =0,
        \quad
        Z\subset X\setminus Y,\,
        Z\neq{\O}.
    \end{equation*}

    The interaction region
    $\Omega_{|X\setminus Y|}(t,\Omega_{|Y|}(0))$
    has a finite volume
    $V_{\Omega_{|X\setminus Y|}}(t)$
    obeying the following estimate \cite{R06SIGMA}
    \begin{equation}
    \label{V<=}
        V_{\Omega_{|X\setminus Y|}}(t)
        \leq
        \Big(C
             +(C_1+C_2t)(s+n)
             +t\sum\limits_{i=1}^{s+n}\,
                   p_i^2
        \Big)^n
        C^s,
    \end{equation}
    where
    $C=\max_{i=1,\ldots,s}\big|l_{i}(0)\big|$,
    $C_1\equiv 2R$,
    $C_2\equiv2\,(4b+1)$
    and
    $b$
    is determined by condition (\ref{Phi<=}).

    By this means, for (\ref{reg<Gt>a}) an inequality
    \begin{multline}
        \label{a<}
        \Big|
            \langle
                G(0)|F(t)
            \rangle
        \Big|
        \leq2\,
        ||{F(0)}||_{L_{\xi,\beta}^{\infty}}\,
        ||G(0)||_{C_\gamma}\,
        \sum\limits_{s=0}^{\infty}\,
            \frac{\big(\xi\mathrm{e}^{2\beta b}\big)^s}{\gamma^s}
            \int\limits_{\Omega_{|Y|}(0)\times\mathbb{R}^s}
                dY\,
                \exp\Big\{-\beta\,\sum\limits_{i=1}^{s}\frac{p_i^2}{2}\Big\}
        \times\\
        \times
        \sum\limits_{n=0}^{\infty}\,
            \frac{\big(2\xi\mathrm{e}^{2\beta b}\big)^n}{n!}
            \int\limits_{\Omega_{|X\setminus Y|}(t)\times\mathbb{R}^n}
                d(X\setminus Y)\,
                \exp\Big\{-\beta\sum\limits_{i={s+1}}^{s+n}\frac{p_i^2}{2}\Big\}
    \end{multline}
    is satisfied.
    In (\ref{a<}) the following estimate
    \begin{equation*}
    \label{U12exp<=}
        \Big(\mathfrak{A}_{1}(-t,Y)
            +\sum\limits_{\substack{Z\subset X\setminus Y\\ Z\neq{\O}}}\,
             \mathfrak{A}_{2}(-t,Y,Z)
        \Big)\,
        \exp\Big\{-\beta\sum\limits_{i=1}^{s+n}\frac{p_i^2}{2}\Big\}
        \leq
        2^{|X\setminus Y|}\,
        \mathrm{e}^{2\beta b |X|}\,
        \exp\Big\{-\beta\sum\limits_{i=1}^{s+n}\frac{p_i^2}{2}\Big\}
    \end{equation*}
    is taken into account. It is a consequence of invariance of Hamiltonian
    of $n$-particle system after acting of the evolution operator
    $S_{n}(-t)$ (\ref{Sspher}),
    boundedness of the interaction region and conditions (\ref{Phi<=})
    on the interaction potential
    $\Phi$
    \cite{CGP97}.

    The expression for estimate (\ref{V<=}) is represented in the following form
    \begin{multline}
    \label{^N}
            \Big(C
                 +(C_1+C_2t)\,s
                 +t\sum\limits_{i=1}^{s}\,
                       p_i^2
                 +(C_1+C_2t)\,n
                 +t\sum\limits_{i=s+1}^{s+n}\,
                       p_i^2
            \Big)^{n}=\\
           =\sum\limits_{k=0}^n\,
                n!\,
            \sum\limits_{l=0}^k\,
                \frac{C^l}{l!}\,
            \sum\limits_{m=0}^{k-l}\,
                \frac{s^m}{m!}\,
                (C_1+C_2t)^m\,
                \frac{t^{k-l-m}}{(k-l-m)!}\,
            \Big(\sum\limits_{i=1}^{s}
                      p_i^2
            \Big)^{k-l-m}
            \times\\
            \times
            \sum\limits_{r=0}^{n-k}\,
                \frac{n^r}{r!}\,
                (C_1+C_2t)^r\,
                \frac{t^{n-k-r}}{(n-k-r)!}\,
                \Big(\sum\limits_{i=s+1}^{s+n}\,
                          p_i^2
                \Big)^{n-k-r}.
    \end{multline}
    After integration with respect to the configuration variables in each term
    of series on the right side of inequality (\ref{a<}), allowing for equality (\ref{^N}) and an inequality as (\ref{pexp<=}) we calculate the integrals over momentum variables. As a result we obtain
    \begin{multline}
    \label{a<<}
            \Big|
                \langle
                    G(0)|F(t)
                \rangle
            \Big|
            \leq 2\,
            ||{F(0)}||_{L_{\xi,\beta}^{\infty}}\,
            ||G(0)||_{C_\gamma}\,
            \sum\limits_{s=0}^{\infty}\,
                \left(\frac{C\xi\mathrm{e}^{2\beta b}}{\gamma}\right)^s\,
                \left(\frac{2\pi}{\beta^{\prime\prime}}\right)^{\frac{s}{2}}\,
            \sum_{n=0}^\infty\,
                \left(2\xi\mathrm{e}^{2\beta b}\right)^n\,
                \left(\frac{2\pi}{\beta^{\prime\prime}}\right)^{\frac{n}{2}}
            \times\\
            \times
            \sum_{k=0}^n\,
                \sum_{l=0}^k\,
                    \frac{C^l}{l!}\,
                    \sum_{m=0}^{k-l}\,
                        \frac{s^m}{m!}\,
                        (C_1+C_2t)^m\,
                        \left(\frac{2t}{\beta^\prime}\right)^{k-l-m}\,
                        \sum_{r=0}^{n-k}\,
                            \frac{n^r}{r!}\,
                            (C_1+C_2t)^r\,
                            \left(\frac{2t}{\beta^\prime}\right)^{n-k-r},
    \end{multline}
    where
    $\beta=\beta^\prime+\beta^{\prime\prime}$.

    Let us set the notation
    $\tilde{C_1}=\max(C_1,1),$
    $\tilde{C_2}=\max\big(C_2,\frac{2}{\beta^{\prime}}\big)$.
    Applying inequalities (\ref{C1+C2t}), (\ref{e^<}) and the relation
    \begin{equation*}
        \sum\limits_{k=0}^n\,
            \big(
                \tilde{C}_1+\tilde{C}_2t
            \big)^k
        \leq
        \big(1
            +\tilde{C}_1
            +\tilde{C}_2t
        \big)^n,
    \end{equation*}
    we see estimate (\ref{a<<}) takes  the form (\ref{A<}).

    From (\ref{A<}) it follows that if
    $\xi$
    complies with condition
    $(\ref{xi<})$,
    then expansion (\ref{reg<Gt>a}) converges under (\ref{t0}).
\end{proof}

For an additive type of observables
$G^{(1)}(0)=\big(0,a_{1}(0,x_1),0,$ $\ldots,0,\ldots\big)$ 
functional (\ref{reg<Gt>a}) becomes
\begin{equation}
\label{G1ta}
       \big\langle
            G^{(1)}
       \big\rangle(t)
      =\big\langle
            G^{(1)}(0)|F(t)
       \big\rangle
      =\int\limits_{\mathbb{R}\times\mathbb{R}}
           dx_1\,
           a_{1}(0,x_1)\,
           F_{1}(t,x_1),
\end{equation}
where 
$F_{1}(t,x_1)$ 
is determined by formula \eqref{FgU2} for 
$|Y|=1.$

Functional (\ref{G1ta}) is well-defined in the case of
\begin{equation*}
\label{xi<}
    \xi
   <\frac{2\mathrm{e}^{-2\beta b-1}}{(2\tilde{C_1}+1)^2-1}\,
    \sqrt{\frac{\beta^{\prime\prime}}{2\pi}}
    \qquad
    \text{as}
    \qquad
    0\leq t<t_0,
\end{equation*}
where
$t_0$
is determined by \eqref{t0}.

\section{Conclusion}
\setcounter{equation}{0}

We prove the local, in time, existence of functionals (\ref{<Gt>b}) and (\ref{<Gt>a}) in  the cases of evolution of observables  and evolution of states for the initial data
$F(0)\in L_{\xi,\beta}^\infty$
describing infinite-particle systems for the specific class of observables
$G(0)\in C_{\gamma,0}$
defined in Section~1.

The results obtained allow us to establish the existence of the mean values for classes of observables
$G(0)\in C_{\gamma,0}$
wider than we considered earlier (e.g., the local observables \cite{Bo48} the mean values of which are governed by the hydrodynamic equations).

Our results can also be applied in the case of the so-called intensive thermodynamic observables (e.g., the number of particles). For the corresponding mean values to be well-defined one must then introduce the densities of these observables.
One can prove the existence of mean values of the intensive thermodynamic observables by using the results obtained above in the thermodynamic limit \cite{CGP97}.

Applying the method of continuation of the BBGKY hierarchy solution developed in \cite{CGP97,Pe79}, one can  prove the global in time  existence of functionals (\ref{<Gt>b}) and (\ref{<Gt>a}) for the initial data close to equilibrium.

Furthermore, the existence of functionals for the characteristics of a deviation from the mean values of observables (e.g., the dispersion) describing fluctuations in non-equilibrium statistical systems can be proved similarly to Propositions~\ref{P:<Gt|F0>} and \ref{P:<G0|Ft>}.

\bigskip
\noindent
\underline{Acknowledgments}:

The research was supported in part by the Ministry of Education and Science of
Ukraine under \linebreak Grant~M/124-2007. I am grateful to Professor V.~I.~Gerasimenko  for  a number of useful comments.

\label{ryabukha:LastPage}
\newpage

\end{document}